\documentclass[10pt, final, journal, letterpaper, twocolumn]{IEEEtran}
\makeatletter
\def\ps@headings{%
\def\@oddhead{\mbox{}\scriptsize\rightmark \hfil \thepage}%
\def\@evenhead{\scriptsize\thepage \hfil \leftmark\mbox{}}%
\def\@oddfoot{}%
\def\@evenfoot{}}
\makeatother 

\IEEEoverridecommandlockouts

\usepackage{bm}
\usepackage{amsfonts}
\usepackage[dvips]{graphicx}
\usepackage{caption}
\usepackage{subfigure}
\usepackage{times}
\usepackage{cite}
\usepackage{lettrine}
\usepackage{amsthm}
\usepackage{amsmath}
\allowdisplaybreaks[4]
\usepackage{array}
\usepackage{amssymb}

\usepackage{stfloats}
\usepackage{slashbox}
\usepackage{graphicx}
\usepackage{footnote}
\usepackage{booktabs}
\usepackage{array}
\usepackage{algorithmic}
\usepackage{algorithm}
\usepackage{subeqnarray}
\usepackage{cases}
\usepackage{threeparttable}
\usepackage{color}

\DeclareMathOperator*{\argmax}{argmax}

\newtheorem{proposition}{\underline{Proposition}}

\IEEEoverridecommandlockouts 
\begin{document}
\bibliographystyle{IEEEtran}

\title{Price-Based Distributed Offloading  for Mobile-Edge Computing with Computation Capacity Constraints}
\author{Mengyu Liu and Yuan Liu,~\IEEEmembership{Member,~IEEE}
\thanks{Manuscript received October 26, 2017; revised November 28, 2017; accepted December 1, 2017.  This paper was supported in part by the Natural Science Foundation of China under Grant 61401159 and Grant 61771203, and in part by the Pearl River Science and Technology Nova Program of Guangzhou under Grant 201710010111. The associate editor coordinating the review of this manuscript and approving it for publication was
M. Nafie.  \emph{(Corresponding author: Y. Liu)}

The authors are  with the School of Electronic and Information Engineering,
South China University of Technology, Guangzhou 510641, China (e-mail: liu.mengyu@mail.scut.edu.cn,  eeyliu@scut.edu.cn).
}
}
\maketitle

\vspace{-1.5cm}

\begin{abstract}
Mobile-edge computing (MEC) is a promising technology  to enable real-time information transmission and computing by offloading computation tasks from wireless devices to network edge.
In this study, we propose a price-based distributed method  to manage the offloaded computation tasks from users.
A Stackelberg game is  formulated to model the interaction between the edge cloud and users, where the edge cloud sets prices to maximize   its revenue  subject to its finite computation capacity, and for given prices, each user locally  makes  offloading decision to minimize its own  cost  which is defined as latency plus payment.
Depending on the edge cloud's knowledge of the network information, we develop the uniform and differentiated  pricing algorithms,  which can both be implemented in  distributed manners.
Simulation results validate the effectiveness of the proposed schemes.

\end{abstract}


\textbf{\textit{Index Terms|}}\textbf{Mobile-edge computing (MEC),  computation offloading, pricing, Stackelberg game.}

\section{Introduction}

With the increasing popularity of new mobile applications of Internet of Things (IoT), future wireless networks with billions of IoT devices are required to support ultra low-latency communication and computing.
However, the IoT devices are usually with small physical sizes and limited batteries, thus always suffer from intensive computation and high resource-consumption for real-time information processing \cite{challenge,save_cloudcomputing,liu_zhang}.

Mobile-edge computing (MEC) is a promising technology to address this problem. Unlike conventional cloud computing integrated with  remote central clouds results in long latency and fragile wireless connections, MEC migrates  intensive   computation tasks from IoT devices to the  physically proximal network edge, and provides low-latency as well as flexible computing and communication services for IoT devices. As a result, MEC is commonly  agreed as a key technology to realize  next-generation wireless networks \cite{latence_mec}.

Joint radio  and computation resource allocation for MEC  has been  recently investigated in  the literature, e.g., \cite{joint_radio_computation,5G_heterogeneous,EH_MEC,chen_game_EE,price_remote-edge,zhangjun_joint,HuangKaiBin_mobileedge,online_EH}.
In general, the MEC paradigm can be divided into two categories: binary offloading \cite{joint_radio_computation,5G_heterogeneous,EH_MEC,chen_game_EE,price_remote-edge} and partial offloading \cite{zhangjun_joint,HuangKaiBin_mobileedge,online_EH}. With binary offloading, the computation tasks at users can not be partitioned but must be executed as a whole either at users or at edge cloud. With partial offloading, the computation tasks at users can be partitioned into different parts for local computing and offloading at the same time.

Among the mentioned literature, a handful of works \cite{chen_game_EE,price_remote-edge} adopted  game theory to design distributed mechanisms  for efficient  resource allocation in MEC systems.
For example, multiuser binary offloading was considered in \cite{chen_game_EE} using a Nash game to maximize the offloaded tasks subject to the total time and energy constraints.  The work \cite{price_remote-edge} considered competition among multiple heterogeneous clouds.
In view of prior works, most of them assumed that the computation capacity of the edge cloud is infinite. However, as  MEC server is located at network edge, its computation capacity should be finite, especially in the networks with  intensive workloads. In this case, a mechanism is needed to control  users' offloaded tasks to make the network feasible. Moreover, the mechanism should provide reasonable incentives for both edge cloud and users to efficiently allocate network resources in a distributed fashion.

Motivated by the above issues, we consider an edge cloud with finite computation capacity, which is treated as a divisible resource to be sold among the users.  The interaction between the edge cloud and users is  modeled as a Stackelberg game, where the edge cloud sets prices  to maximize its revenue and each user  designs the offloading decision individually to minimize its cost that is defined as latency plus payment.
The main contributions are two-fold:
  1) We propose a new game-based distributed  MEC scheme, where  the users  compete for the edge cloud's finite computation resources via a pricing approach, which is  modeled as a Stackelberg game.
  2) The optimal uniform and differentiated pricing algorithms  are proposed, which can   be implemented with a distributed manner.

\section{System Model and Problem Formulation}

\subsection{System Model}

We consider a   MEC  system with  $K$  users and one base station (BS) that is integrated with a MEC server to execute   the offloaded data of the users. All nodes have a single antenna. The users' computation data   can be arbitrarily divided in bit-wise for  partial local computing and partial offloading.
We assume that the total bandwidth $B$ is equally divided for $K$ users such that each user can occupy a non-overlapping frequency to offload its data to the edge cloud simultaneously.  The quasi-static channel model is considered, where channels remain unchanged during each offloading period, but can vary in different offloading periods.  We also assume that the computation offloading can be completed in a period.

Let $C_k$ denote the number of CPU cycles  for computing 1-bit of input data at user $k$. It is assumed that  user $k$ has to execute $R_k$ bits of input data in total, where $0\leq \ell_k \leq R_k$ bits are offloaded to the edge cloud while the rest $(R_k-\ell_k)$ bits are computed by its local CPU. The local CPU frequency of user $k$ is denoted as $F_k$ that is measured by the number of CPU cycles per second. Then the time for local computing at user $k$ is $t_{\rm{loc},\textit{k}}=(R_k-\ell_k)C_k/F_k$.
The offloading time $t_{{\rm off},k}$ of user $k$ comprises three parts: the uplink transmission time $t_{{\rm u},k}$, the execution time at the cloud $t_{{\rm c},k}$, and the downlink feedback time $t_{{\rm d},k}$. Thus, the offloading time $t_{{\rm off},k}$ is
\begin{align}
t_{{\rm off},k} = t_{{\rm u},k} + t_{{\rm c},k} + t_{{\rm d},k}.
\end{align}
Since the local computing and offloading can be performed concurrently, the required time of user $k$ for  executing  the total $R_k$ bits data can  be expressed as
$
t_k=\max \{t_{\rm{loc},\textit{k}},t_{\rm{off},\textit{k}}\}.
$

More specially,  the data size of the computed result fed back to user $k$ is  $\alpha_k \ell_k$,   where  $\alpha_k$ $ (\alpha_k >0)$ accounts for the ratio of output to input bits offloaded to the cloud \cite{time_beta}, which depends on the applications of the users.  Then we have $t_{{\rm u},k} =  \ell_k / r_k$ and $t_{{\rm d},k} = \alpha_k \ell_k / r_{{\rm B},k}$, where $r_k = \frac{B}{K}\log_2\left(1+\frac{p_k h_k}{B/K N_0}\right)$ and $r_{{\rm B},k} = \frac{B}{K}\log_2\left(1+\frac{P_{{\rm B},k} h_k}{B/K N_0}\right)$ denote the uplink and downlink transmission rates for user $k$, respectively.  Here $N_0$ is the noise power spectrum density, $h_k$ is the channel gain between the BS and user $k$, and $P_{{\rm B},k}$ and $p_k$ are the downlink and uplink power for user $k$, respectively.
Moreover, let $f_{{\rm c},k}$  denote the computational speed of the edge cloud assigned to user $k$, then we have $t_{{\rm c},k} = \ell_k C_k / f_{{\rm c},k}$. Here we consider equal $f_{{\rm c},k}$ allocation for simplicity, i.e., $f_{{\rm c},k} = f_C/K$,  where $f_C$  is the total computational speed of the could.

We consider a practical constraint that  the edge cloud has finite computation capacity so that its CPU cycles for computing the sum received  data  in each offloading period are upper bounded by $\overline{F}$,  the  constraint can be expressed as
\begin{align} \label{finitecloud}
\sum_{k=1}^{K}\ell_k C_k \leq \overline{F}.
\end{align}
Note that $\overline{F}$ and $f_C$ represent the MEC server's computational quantity and speed for the offloaded CPU cycles, respectively.

\subsection{Stackelberg Game Formulation}

In this paper, the users consume the edge cloud's  resources to execute the computation tasks while the edge cloud has to ensure  its available  CPU cycles for computing the total offloaded data to be below the computation  capacity. Hence, to adjust the demand and supply of the computation resources, it is considered that the edge cloud prices the CPU cycles $\ell_k C_k$ of the offloaded data for each user $k$. Thus the Stackelberg game can be applied to model the interaction  between the edge cloud and  users,  where the edge cloud is the leader and the users are the followers. The edge cloud (leader) first imposes the prices for CPU cycles of users. Then, the users (followers) divide their input data for local computing and offloading  individually based on the prices announced from the edge cloud.

Denote the CPU cycle prices for  users as a set $\bm{\mu}=\{ \mu_1,\cdots,\mu_K \}$. The objective of the edge cloud is to maximize its revenue obtained from selling the finite computation  resources  to users. Mathematically, the optimization problem at the edge cloud's side can be expressed as (leader problem)
\begin{align}
\textbf{P1}: \quad \max_{\bm{\mu}\succeq 0} \quad U_B(\bm{\mu}) = \sum_{k=1}^{K} \mu_k \ell_k C_k
 \quad \rm{s.t.} \quad \eqref{finitecloud}. \nonumber
\end{align}
Note that the offloaded data $\ell_k$ for user $k$ is actually a function of $\mu_k$, since the size of  data that each user is willing to offload is dependent on its assigned price.

At the users' side,  each user's cost is defined as its latency plus the payment  charged by the edge cloud, i.e.,
\begin{align}
U_k(\ell_k,\mu_k) =  t_k + \mu_k \ell_k C_k,
\end{align}
which is equivalent to
\begin{align} \label{utilityfunction}
U_k(\ell_k,\mu_k)=
\begin{cases}
\big(\mu_k - \frac{1}{F_k}\big)\ell_k C_k + \frac{R_k C_k}{F_k},       &0\leq \ell_k \leq m_k,\\
\beta_k \ell_k + \mu_k \ell_k C_k,   &m_k < \ell_k \leq R_k,
\end{cases}
\end{align}
where $\beta_k = \frac{1}{r_k} + \frac{C_k}{f_{{\rm c},k}} + \frac{\alpha_k}{r_{{\rm B},k}}$ and $0< m_k< R_k$ is defined as $m_k = \frac{ C_k R_k }{\beta_k F_k +  C_k }$.


The goal  of each user $k$ is to minimize its own cost by choosing the optimal offloaded data size $\ell_k$ for given price $\mu_k$ set by the edge cloud. Mathematically, this problem can be expressed  as (follower problem)
\begin{align}
\textbf{P2}: \quad \min_{\ell_k} \quad U_k(\ell_k,\mu_k)
 \quad {\rm{s.t.}} \quad 0\leq \ell_k \leq R_k. \nonumber
\end{align}
It is worth noting that the payment term in Problem \textbf{P1} and Problem \textbf{P2} can be cancelled each other from the net utility perspective.  Problem \textbf{P1} and Problem \textbf{P2} in the Stackelberg game are coupled in a complicated way, i.e., the pricing strategies of the edge cloud have an influence on the offloaded data sizes of the  users which also impact the edge cloud's revenue in turn.

\section{Optimal Algorithm}

To analyze  the considered Stackelberg game, each user $k$ independently decides its offloading strategy $\ell_k^*$ by solving Problem \textbf{P2} with given  price  $\bm{\mu}$. Knowing each user's offloading decision $\ell_k^*(\mu_k)$, the edge cloud sets its optimal price $\bm{\mu}^*$ by solving Problem \textbf{P1}. The above process is known as the backward induction. In this paper, two optimal pricing strategies are considered, which are termed as uniform pricing and differentiated pricing \cite{d2d_game}. In the following, we will investigate the two pricing schemes respectively.

\subsection{Uniform Pricing}

For the uniform pricing scheme, the  edge cloud sets and broadcasts a uniform price to all  users, i.e., $\mu = \mu_1\cdots=\mu_K$. For given uniform price $\mu$, the objective function $U_k$ is a piecewise function of $\ell_k$, which is linear in each interval from  \eqref{utilityfunction}.
Then, by exploiting the structure  of  $U_k$, we can  obtain the optimal solution for Problem \textbf{P2} in the following proposition.
\begin{proposition} \label{threshold-based}
The optimal offloading decision  of each user  in Problem \textbf{P2} follows the threshold-based policy, i.e.,
\begin{align} \label{Usolution_P2}
\ell_k^* (\mu)= m_k x_k, \quad \forall k,
\end{align}
where the binary variable $x_k$ is defined as
\begin{eqnarray} \label{U_binary}
x_k=
\begin{cases}
1,       &\mu \leq \frac{1}{F_k},\\
0,   & \mbox{otherwise}.
\end{cases}
\end{eqnarray}
\end{proposition}
\begin{proof}
Please refer to Appendix A.
\end{proof}

From Proposition \ref{threshold-based}, we obtain that the offloading threshold is $1/\mu$, i.e., user $k$ prefers to offload $m_k$ bits to the edge cloud if its CPU frequency $F_k$ is smaller than or equal to the threshold, and leaves all bits for local computing otherwise.
In other words, the computation offloading is beneficial if the user has small computational speed $F_k$ and it is likely to compute locally otherwise.
%

Then we turn our attention to Problem \textbf{P1}. By substituting \eqref{Usolution_P2} into Problem \textbf{P1}, the optimization problem at the edge cloud for the uniform pricing scheme can be rewritten  as
\begin{align}
\textbf{P3}: \quad \max_{\mu \geq 0} &\quad U_B(\mu) = \mu \sum_{k=1}^{K}  m_k x_k C_k \label{revenue_U}\\
 \rm{s.t.} &\quad \sum_{k=1}^{K} m_k x_k C_k \leq \overline{F}. \label{U_cloud_capacity}
\end{align}
%
%
\begin{proposition} \label{uniform_finite_set}
Without loss of generality, we sort  $1/F_1<\cdots<1/F_{K-1}<1/F_K$,  the optimal uniform price $\mu^*$ must belong to the set
$\{1/F_1,\cdots,1/F_{K-1},1/F_K\}. $
\end{proposition}
\begin{proof}
Please refer to Appendix B.
\end{proof}


According to Proposition \ref{uniform_finite_set}, the revenue maximization problem in  \textbf{P3} reduces to the one-dimensional search problem over $K$ values in $\{1/F_k\}_{k=1}^K$ and we summarize the whole method in Algorithm 1  formally.
Specially, the edge cloud bargains with the users  by announcing price   in the decreasing order of $\{1/F_k\}_{k=1}^{K}$. Since the required sum CPU cycles  $\sum_{k=1}^{K}\ell_k C_k$ decreases with the price $\mu$, the price bargaining ends as long as  the computation  capacity constraint \eqref{U_cloud_capacity}  is active  and there is no need for bargaining the rest of price candidates.

It is obvious that  the total  complexity of Algorithm 1 to search $\mu^*$ is $\mathcal{O}(\log K)$. For the uniform pricing scheme, the edge cloud needs the limited network information, i.e., $F_k$ and $C_k$, which are collected by the cloud before the algorithm. In each iteration, after knowing the price $\mu$ broadcasted by the edge cloud, each user independently makes its offloading decision $\ell_k$ and reports it to the edge cloud for updating the price. Therefore, the cloud broadcasts the price $\mu$ and each user reports its offloading decision $\ell_k$, which are the information exchanged between the edge cloud and the users in each iteration. Hence, Algorithm 1 is a fully distributed algorithm.

\begin{algorithm}[tb]
\caption{Optimal Uniform Pricing Policy for Problem \textbf{P3}}
\begin{algorithmic}[1]
\STATE The edge cloud initializes $\tau = K$ and $\mu^{\tau}=1/F_{\tau}$.
\REPEAT
\STATE Every user decides its optimal offloaded data size $\ell_k^*(\mu^{\tau})$ according to \eqref{Usolution_P2}.
\STATE The edge cloud  computes its revenue $U_B(\mu^{\tau})$ from \eqref{revenue_U}.
\IF{$\sum_{k=1}^{K}\ell_k^*(\mu^{\tau})C_k \leq \overline{F}$}
\STATE Update the price $\mu^{\tau-1}=1/F_{\tau-1}$, and $\tau\leftarrow \tau - 1$;
\ELSE
\STATE Set $U_B(\mu^{\tau})=0$; \textbf{beak};
\ENDIF
\UNTIL $\tau \leq 0$.
\STATE Output $\mu^*\leftarrow\argmax_{\mu^{\tau}}U_B(\mu^{\tau})$.
\end{algorithmic}
\end{algorithm}

\subsection{Differentiated Pricing}

Here, we consider the general case where the edge cloud charges the different users with different prices. Similar to the uniform pricing case, the optimal solution for Problem \textbf{P2} is also \eqref{Usolution_P2}, except that $\mu$ in $x_k$ is replaced by  $\mu_k$.
And, Problem \textbf{P1} can be rewritten as
\begin{align}
\textbf{P4}: \quad \max_{\bm{\mu} \succeq 0} \quad U_B(\bm{\mu}) = \sum_{k=1}^{K} \mu_k m_k x_k C_k
\quad \rm{s.t.} \quad \eqref{U_cloud_capacity}.\nonumber
\end{align}
%

It is worth noting that the price $\mu_k$ is actually a function of $x_k$ for user $k$. Specifically,   $x_k=1$, i.e., $\mu_k\leq 1/F_k$ and the optimal price for user $k$ is thus given by $\mu_k^*=1/F_k$ as the objective function of Problem \textbf{P4}  is an increasing function of $\mu_k$. When $x_k=0$, the edge cloud sets the price for user $k$ as $\mu^*_k=\infty$ and earns no revenue. Based on the above analysis, Problem \textbf{P4} is thus equivalent to
\begin{align}
\textbf{P4$'$}: \quad \max_{x_k \in \{0,1\}} \quad U_B(x_k) = \sum_{k=1}^{K} \frac{ m_k x_k C_k}{F_k}
 \quad \rm{s.t.} \quad \eqref{U_cloud_capacity}. \nonumber
\end{align}
Problem \textbf{P4$'$} is actually a \emph{binary knapsack problem}
with the weight $m_k C_k$ and the value $m_k C_k/F_k$ for user $k$.
Since the problem is  NP-complete, there is no efficient algorithm solving it optimally. However, we can apply dynamic programming \cite{DY} to solve the above binary knapsack problem in pseudo-polynomial time.\footnote{
We adopt kp01$(\cdot)$ software package  in MATLAB.}

Each user $k$ needs to report  $m_k,$ $C_k,$ and $F_k$ to the edge cloud for solving Problem \textbf{P4$'$} and there is no need for iteration between the edge cloud and users.
Obtaining the optimal price $\mu_k^*$, user $k$ decides its optimal strategy   based on \eqref{Usolution_P2}. Therefore,  the differentiated pricing scheme  is also a distributed algorithm, but it needs more information and higher complexity than that of the uniform pricing scheme.

\section{Numerical Results}


In the simulation setup, we assume that the total channel bandwidth $B$ is 1 MHz and the noise power spectrum density $N_0$ is $-174$ dBm/Hz.
Each $h_k$ is assumed to be uniformly distributed in $[-50,-30]$ dBm.
The local CPU frequency $F_k$ for each user $k$ is uniformly selected from the set $\{0.1,0.2,\cdots,1\}$ GHz, and the required number of CPU cycles per bit and  the data size for user $k$ are  uniformly  distributed  with $C_k\in[500,1500]$ cycles/bit and $R_k\in [100,500]$ KB, respectively.
Unless otherwise noted, the remaining parameters are set as follows: $p_k = 0.1$ W, $P_{{\rm B},k} = 1$ W, $\overline{F}=6\times10^9$ cycles/slot,  $\alpha_k = 0.2$, and $f_C=100$ GHz.

The average performance of the two proposed pricing schemes is evaluated and compared in terms of average latency and revenue. Besides, we consider the scheme where all input data is computed locally at users for comparison.
In Fig. \ref{big:figure}(a),  both latency and revenue performance become better as the computation capacity increases, while the scheme of only local computing has the worst latency performance and is not related with the  computation capacity.
In addition, the differentiated pricing scheme has  better performance in both  latency and revenue, which shows that it is more accurate to allocate resource. Thus there exists a tradeoff between performance and complexity for the two pricing schemes.

Fig. \ref{big:figure}(b) illustrates the effect of the number of users on the average latency and  revenue, and we have the similar observations as Fig. \ref{big:figure}(a) for the three schemes. Besides, with the increasing number of users, the allocated spectrum for each user decreases, resulting in lower transmission rate and thus higher latency.  Moreover, it is expected that competition with more users forces up the prices and revenue of the edge cloud.

\begin{figure}
\centering
\subfigure[ $K=30$]{
\begin{minipage}[b]{0.22\textwidth}
\includegraphics[width=1\textwidth]{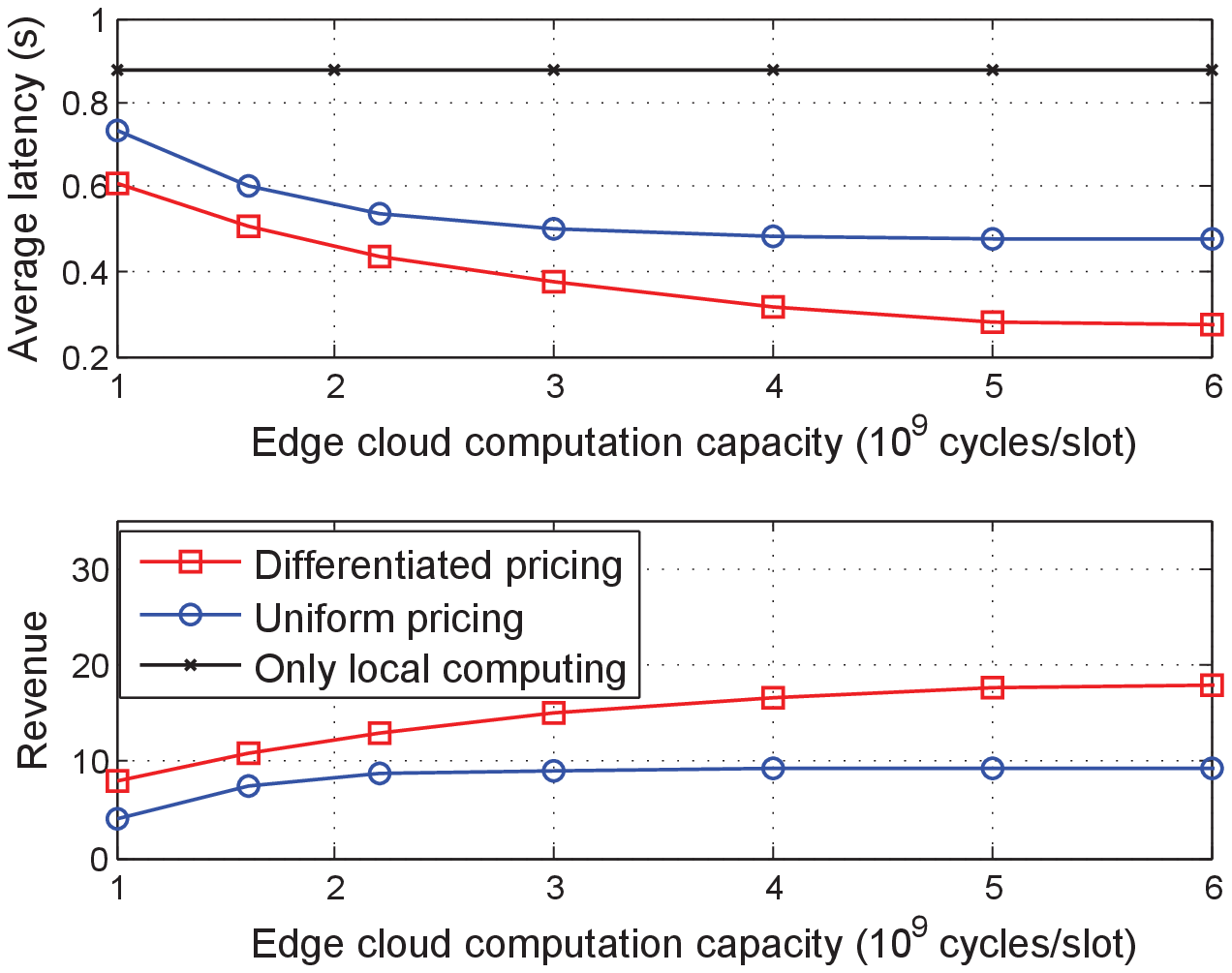}
\end{minipage}
}
\subfigure[ $\overline{F}=6\times10^9$ cycles/slot]{
\begin{minipage}[b]{0.22\textwidth}
\includegraphics[width=1\textwidth]{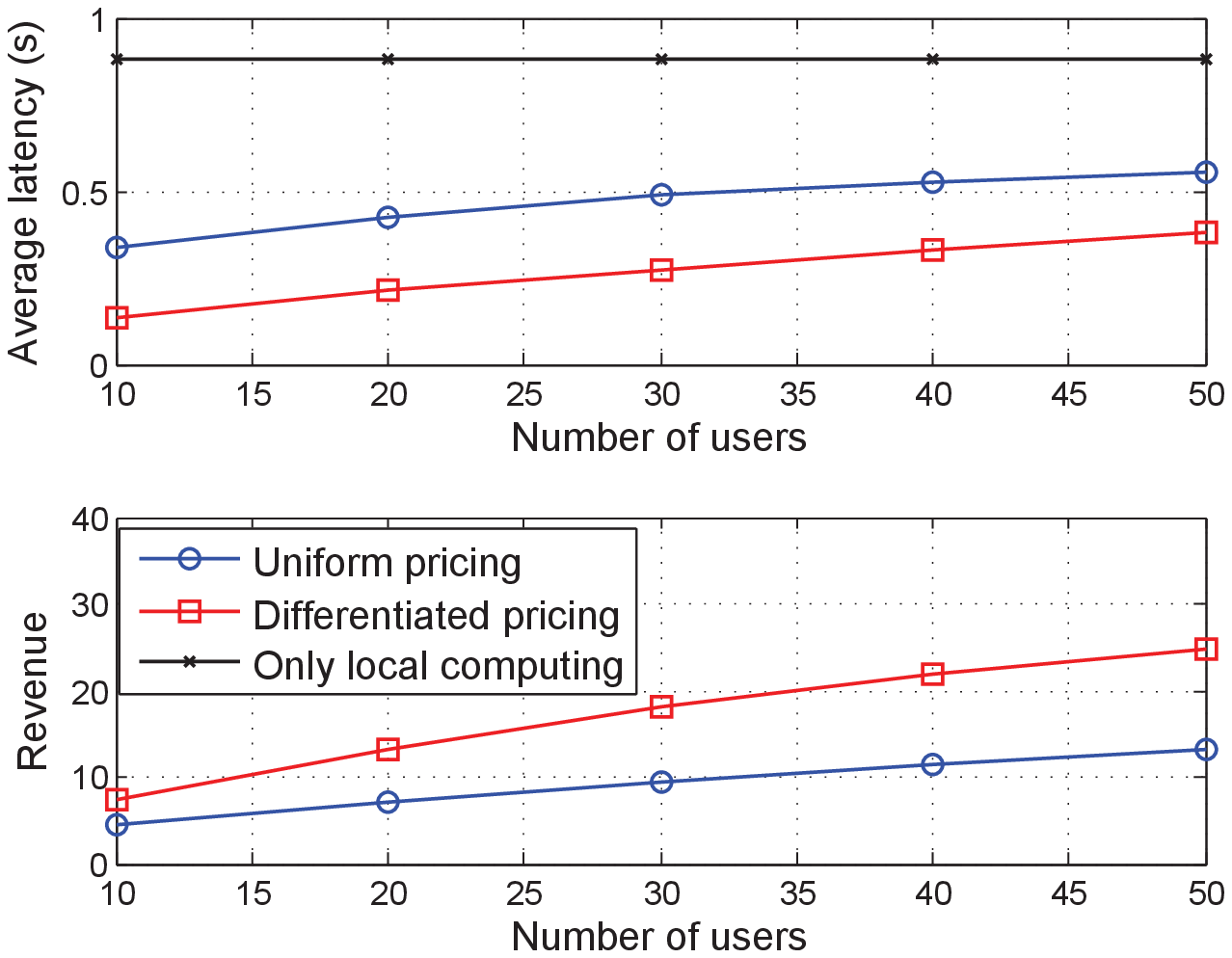}
\end{minipage}
}
\caption{Performance comparison.} \label{big:figure}
\end{figure}

\section{Conclusion}

In this work, we investigated the price-based computation offloading  for a multiuser MEC system. The finite computation capacity of the edge cloud was considered to manage the offloaded tasks from users. The Stackelberg game was applied to model the interaction between the edge cloud and  users.
Based on the edge cloud's knowledge of the network information, we proposed  uniform and differentiated pricing schemes, which can both be implemented with distributed manners.

\appendices

\section{The proof of Proposition \ref{threshold-based}}

For  given  $\mu$, the optimal solution for Problem \textbf{P2} is
%
\begin{eqnarray} \label{U_l}
\ell_k^*(\mu)
\begin{cases}
=m_k,       &\mu < \frac{1}{F_k},\\
\in[0,m_k],       &\mu = \frac{1}{F_k},\\
=0,   & \mbox{otherwise}.
\end{cases}
\end{eqnarray}
%
The  case $\mu = 1/F_k$ for all $k$ occurs with probability $0$ and we let $\ell_k=m_k$ in this case,  which completes the proof.

\section{The proof of Proposition \ref{uniform_finite_set}}

It can be proved by contradiction as follows.
Suppose that the optimal price $\mu^*$ can exist in the interval  $(1/F_i , 1/F_{i+1})$,  $\forall i \in \{1,\cdots,K-1\}$. Then, we consider the case $\widetilde{\mu}=1/F_{i+1}$.
From \eqref{U_binary}, we can obtain that $x_k=0$ with $k=1,\cdots,i$ and  $x_k=1$ with $k=i+1,\cdots,K$ for both $\widetilde{\mu}$ and $\mu^*$. Therefore, the CPU cycles of the sum offloaded data $\sum_{k=1}^{K}m_k x_k C_k$ for $\widetilde{\mu}$ are equivalent to that for $\mu^*$. As the objective function $U_B(\mu)$ given in \eqref{revenue_U} is an increasing linear function of the price $\mu$, we can always have that the case $\widetilde{\mu}=1/F_{i+1}$ achieves a higher revenue than the case $\mu^*$. Thus, this contradicts with the assumption that $\mu^*$ is optimal for Problem \textbf{P3} with $1/F_i < \mu^* < 1/F_{i+1}$. Therefore, the optimal price $\mu^*$ must exist in the set $\{1/F_1,...,1/F_{K-1},1/F_K\}$.

\bibliography{Liu_WCL2017_1113.R1}
\end{document}